\documentclass[journal]{IEEEtran}

\ifCLASSINFOpdf

\else

\fi

\usepackage{amssymb}
\usepackage{amssymb}
\usepackage{amsmath}
\usepackage{mathrsfs}
\usepackage{amsmath, amssymb}
\usepackage{amsmath}
\usepackage{graphicx}
\usepackage{amsmath, amssymb, amsthm}
\usepackage{leftidx}
\usepackage{extarrows}
\usepackage{stfloats}
\usepackage{floatrow}
\newtheorem{prop}{Proposition}

\newtheorem{theorem}{Theorem}

\newtheorem{rmk}{Remark}
\newtheorem{lem}{Lemma}
\newtheorem{problem}{Problem}

\newtheorem{defnn}{Definition}

\usepackage{amssymb}
\usepackage{picinpar}
\usepackage{epstopdf}
\usepackage[nosort]{cite}
\usepackage{algorithm}
\usepackage{algpseudocode}
\usepackage{caption}
\usepackage{subfigure}
\begin{document}

\title{Categorization Problem on Controllability of\\ Boolean Control Networks}

\author{Qunxi~Zhu,~Zuguang Gao,~Yang~Liu,~\IEEEmembership{Member IEEE},~Weihua Gui
\thanks{Qunxi Zhu's research was sponsored by the China Scholarship Council. Zuguang Gao's research was supported in part by Chicago Booth Ph.D. fellowship and Oscar G. Mayer Ph.D. fellowship. This work was also supported in part by the National Natural Science Foundation of China under grant 61321003, and the Natural Science Foundation of Zhejiang Province of China under grant D19A010003. ({\textit{Corresponding author: Yang Liu.}})
}

\thanks{Q. Zhu is  with the School of Mathematical Sciences, Fudan University, Shanghai 200433, China. (e-mail:qxzhu16@fudan.edu.cn)}
\thanks{Z. Gao is with Booth School of Business, University of Chicago, Chicago, IL 60637, United States. (e-mail:zuguang.gao@chicagobooth.edu)}
\thanks{Y. Liu is with the School of Mathematics and Computer Science, Zhejiang Normal University, Jinhua 321004, China. (e-mail:liuyang@zjnu.edu.cn)
}
\thanks{W. Gui is with the School of Information Science and Engineering, Central South University, Changsha 410083, China. (e-mail:gwh@csu.edu.cn)}
}

%\author{Qunxi~Zhu,~Yang~Liu,~\IEEEmembership{Member IEEE},~Jianquan~Lu,~\IEEEmembership{Member IEEE},~Jinde~Cao,~\IEEEmembership{Fellow IEEE}
%\thanks{This work was supported by the National Natural Science Foundation of China under Grant No.  11671361, {61573096} and 61573102, the Natural Science Foundation of Jiangsu Province of China under Grant BK20170019, and China Postdoctoral Science Foundation under Grant No. 2016T90406,  2015M580378, 2014M560377, and 2015T80483, Jiangsu Province Six Talent Peaks Project under Grant 2015-ZNDW-002, {and the Jiangsu Provincial Key Laboratory of Networked Collective Intelligence under Grant No. BM2017002. ({\textit{Corresponding author: Yang Liu.}})}}
%
%\thanks{Q. Zhu is  with the School of Mathematical Sciences, Fudan University, Shanghai 200433, China. (e-mail:qxzhu16@fudan.edu.cn)}\thanks{Y. Liu is with the College of Mathematics,  Physics and Information
%Engineering,  Zhejiang Normal University, Jinhua 321004, China, and also with the School of Mathematics, Southeast University, Nanjing 210096, China. (e-mail:liuyang@zjnu.edu.cn)
%}
%\thanks{J. Lu and J. Cao   are with the School of Mathematics, Southeast University,
%Nanjing 210096, China. (e-mail:jqluma@seu.edu.cn; jdcao@seu.edu.cn)}
%}

\maketitle

\begin{abstract}
    A Boolean control network (BCN) is a discrete-time dynamical system whose variables take values from a binary set $\{0,1\}$. At each time step, each variable of the BCN updates its value simultaneously according to a Boolean function which takes the state and control of the previous time step as its input. Given an ordered pair of states of a BCN, we define the \emph{set of reachable time steps} as the set of positive integer $k$'s where there exists a control sequence  such that the BCN can be steered from one state to the other in exactly $k$ time steps; and the \emph{set of unreachable time steps} as the set of $k$'s where there does not exist any  control sequences such that the BCN can be steered from one state to the other in exactly $k$ time steps. We consider in this paper the so-called \emph{categorization problem} of a BCN, i.e., we develop a method, via algebraic graph theoretic approach, to determine whether the set of reachable time steps and the set of unreachable time steps, associated with the given pair of states, are finite or infinite. Our results can be applied to classify all ordered pairs of states into four categories, depending on whether the set of reachable (unreachable) time steps is finite or not.
\end{abstract}


\IEEEpeerreviewmaketitle
\section{Introduction}

The Boolean network (BN) was firstly proposed by Kauffman \cite{aa1} to model  gene regulatory networks (GRNs). BN is a simple yet quite powerful tool for analizing GRNs, compared with other tools such as those involving ordinary differential equations, which often have numerous unknown parameters and can be hardly solved for large-scale systems~\cite{smolen2000mathematical}.
In addition, the BNs facilitate to study the possible steady-state behaviors systematically.   For example, Albert \textit{et al.} proposed a simplified BN of the segment polarity gene network of \textit{Drosophila} melanogaster \cite{albert2003topology}. Such a BN can provide an essential qualitative description for the expression of genes.  BNs with external control inputs are called Boolean control networks (BCNs). A typical example is the cell cycle control network of fission yeast~\cite{davidich2008boolean}.

In the past decade, Cheng and his colleagues \cite{Cheng2010} have proposed a seminal technique, called semi-tensor product (STP) of matrices, for analyzing BNs and BCNs. Some applications of STP include the analysis of controllability~\cite{Cheng2009controllability,aa18,Laschov2012,zhu2019further}, observability~\cite{Cheng2009controllability,Fornasini2013,cheng2016note,Zhu2018OB,guo2018redefined}, stability and stabilization~\cite{chen2018monostability,guo2018stability,meng2019stability,li2017lyapunov,li2017stabilization},  optimal control~\cite{aa9,aa10,wu2019optimal} %, disturbance decoupling problem \cite{Liu2017Pinning, Yu2018Block}
    and so on. Moreover, other kinds of BNs and BCNs, such as the conjunctive Boolean networks (CBNs)~\cite{gao2018controllability,weiss2018minimal,weiss2018polynomial,gao2018stability,chen2017asymptotic}, are recently prevalence. It is no surprise that the research on the BNs and BCNs has become increasingly attractive and challenging. Specifically, the study of controllability has developed rapidly in recent years~\cite{Cheng2009controllability,aa18,Laschov2012,zhu2019further}. One of the most influential results on controllability was provided in~\cite{aa18}, where they defined a so-called controllability matrix, and the controllability of the BCN can be determined by checking the positiveness of the controllability matrix. Additionally, Laschov and Margaliot~\cite{Laschov2012} further studied the $k$~fixed-time controllability by applying the Perron-Frobenius theory. Roughly speaking, an ordered pair of states is $k$~fixed-time controllable if there exists a control sequence that drives the system from the first state to the second state in exactly~$k$ time steps. The results in~\cite{Laschov2012} relates the~$k$ fixed-time controllability with the  positiveness and primitivity of some matrices.
    We will formally define these concepts and introduce the relevant results in section~\ref{section2}.

    In this paper, we propose and answer the following questions: Given a starting state and an ending state, is there infinite number of positive integer~$k$'s such that the pair is~$k$ fixed-time controllable? Is there infinite number of positive integer~$k$'s such that the pair is not~$k$ fixed-time controllable? Equivalently, we define the set of reachable time steps (set of unreachable time steps, respectively) as the collection of positive integer $k$'s such that the given pair of states is~$k$ fixed-time controllable (not~$k$ fixed-time controllable, respectively), and check the finiteness of these two sets. A complete answer to this question is provided as Theorem~\ref{thm:main02}, and some further result is also presented (see Theorem~\ref{lem:Condensationmatrix}).

    The motivation of our study is two-fold. First, we note that a BCN is said to be~$k$ fixed-time controllable if every ordered pair of states of the BCN is~$k$ fixed-time controllable. It was shown in~\cite{Laschov2012} that if a BCN is~$k$ fixed-time controllable, then the BCN is also~$p$ fixed-time controllable for any $p\ge k$ (see Theorem~\ref{las_2} in section~\ref{section2}). However, for a specific \emph{pair of states} which is~$k$ fixed-time controllable, it is not necessarily true that the \emph{pair} is~$p$ fixed-time controllable for any $p\ge k$. A natural question one may ask is that for a given pair of states, does there exist some integer~$k$ such that the pair is~$p$ fixed-time controllable for any $p\ge k$. If the answer is yes, we say that this pair of states falls into the \emph{primitive} category. If the answer is no, we further classify those pairs into three other categories. The detailed formulation is provided in section~\ref{section2}.

    A second motivation of our research comes from potential biological applications. The goal of interest may be to drive a system from one state to another, assuming that the former is undesired and the latter is desired. Additionally, one may encounter the situation that a biological system consists of several identical subsystems with no couplings among them, and each subsystem is modeled by the same BCN. For example, a multi-cellular organism has identical BCNs, each modeling a cell-cycle~\cite{Laschov2012}. We may be interested in finding a control law with respect to each subsystem to drive each subsystem from different initial states to the same desired state at some fixed time. Our results in this paper characterize all possible values of such fixed times efficiently, without checking each positive integer. If such a fixed time exists, all subsystems can be applied with the same control law afterwards, resulting in a complete synchronization of the states of these subsystems in the following dynamical evolutions.

    The remainder of this paper is organized as follows. Section~$\ref{section2}$ introduces some preliminaries on algebraic graph theory and the existing controllability results of BCNs. In section~$\ref{section3}$ we present the categorization problem on controllability of BCNs and establish our main result. An illustrative example is provided in section~$\ref{section4}$. Finally, we conclude the paper in section~\ref{sec5}.

    Before ending this section, we present the following notations that will be used throughout the paper: $\mathbb{Z}^+$ -- the set of the positive integers; $[a,b]$ --  the integer set $\{a,a+1,...,b\}$ with  $a\leq b$; ${\rm Col}_i(A)$ -- the $i$th column of the matrix $A$; $\Delta_k:=\{\delta_k^i~|~i=1, 2,\cdots, k\}$, where $\delta_k^i$ is the $i$th column of the identity matrix $I_k$; $\mathscr{D}:=\{T=1,F=0\}$ -- the logic field; An $m\times n$ matrix $A$ with ${\rm Col}_i(A)\in\Delta_{m}$ for all $i$ -- the logical matrix;  $\mathcal{L}_{m\times n}$ -- the set of all $m\times n$ logical matrices; $A = \delta _m[i_1,i_2,...,i_n]$  -- the simplified expression for $A = [\delta _m^{i_1},\delta _m^{i_2},...,\delta _m^{i_n}] \in \mathcal{L}_{m \times n}$; ${\mathscr{B}}_{n\times n}$ -- the set of $n\times n$  Boolean matrices, i.e., all entries are $0$ or $1$;
    $\mathscr{B}(A)$ --- Boolean form of nonnegative matrix $A$, which is a Boolean matrix with the $ij$th entrie $1$ if $A_{ij}>0$, and the $ij$th entrie $0$ if $A_{ij}=0$.
    $A+_{\mathscr{B}}B=(A_{ij}\vee B_{ij})$ (resp. $A{\times}_{\mathscr{B}}B:=\left(\sum\limits_{k = 1}^{n }{}_{{\mathscr{B}}} (A_{ik}\wedge B_{kj})\right)_{ij}$) -- the Boolean addition (resp. product) of $A, B\in {\mathscr{B}}_{n\times n}$;  $A^{(k)}:=\underbrace {A{ \times _{\mathscr{B}}} \cdots { \times _{\mathscr{B}}}A}_k$; A matrix $A > 0$ means its entries are positive; $|C|$ -- the cardinal number of the set $C$. $\ltimes$ -- semi-tensor product (STP) of matrices.

\section{Problem Formulation and Backgrounds}\label{section2}

    %In this paper, solving Problem \ref{prob},  i.e., obtaining the controllability category matrix $\mathcal{C}$, is mainly based on the properties of the reducible matrix. To this end,  we first present some basic preliminaries about algebraic graph theory \cite{brualdi1991combinatorial}.

\subsection{Problem formulation}
In this subsection, we formally introduce the categorization problem. We first need the following definitions.

A BCN with $n$ state variables can be described as follows:
\begin{equation} \label{eq:1}
        x_i(t+1)=f_i(x_1(t),...,x_n(t); u_1(t),...,u_m(t)),\\
\end{equation}
where $x_i \in\mathscr{D}, i\in[1, n]$ are the state variables, $u_i\in\mathscr{D}^{m}, i\in[1,m]$ are the input variables, and $f_i:\mathscr{D}^{n+m}\rightarrow\mathscr{D}$, $i\in [1,n]$ are the logical functions.
With vector form expression, i.e., we use $\delta_2^1$ to represent state~$1$ and $\delta_2^2$ to represent state~$0$,
one has $x_i, u_i\in \Delta_2$. Then as in~\cite{Cheng2010},~\eqref{eq:1} can be transformed into the algebraic form:
\begin{equation}\label{eq:2}
    x(t+1)=L\ltimes u(t) \ltimes x(t),
\end{equation}
where $x(t)=\ltimes_{i=1}^{n}x_{i}(t) \in \Delta_N$ with $N:=2^n$, $u(t)=\ltimes_{j=1}^{m}u_{j}(t)\in \Delta_M$ with $M:=2^m$, and $L \in \mathcal{L}_{N\times NM}$. Let
\begin{equation}\label{controllabiiltymtraix}
    \mathcal{M}:=\sum\limits_{j=1}^{M}{}_{{\mathscr{B}}}~L\ltimes \delta_M^j~\mbox{and}~
    \mathcal{F}:=\sum\limits_{i=1}^{N}{}_{{\mathscr{B}}}~\mathcal{M}^{(i)},
\end{equation}
where $\mathcal{F}$ is called the \emph{controllability matrix} \cite{aa18}.
We define the controllability of a BCN as follows.
\begin{defnn}[Controllability~\cite{aa18,Laschov2012}]
    The BCN~\eqref{eq:1} is
    \begin{enumerate}
        \item controllable  from  $x_0$ to $x_d$, if there are a $T > 0$ and a sequence of control  $u(0)$,...,$u(T - 1)$, such that driven by these controls the trajectory can go from $x(0) = x_0$ to $x(T) = x_d$;
        \item controllable at $x_0$, if it is controllable from $x_0$ to destination $x_d = x$, $\forall x$;
        \item controllable, if it is controllable at any $x$.
    \end{enumerate}
\end{defnn}
We also define the $k$ fixed-time controllability of a BCN.
\begin{defnn}[$k$ fixed-time controlllability~\cite{Laschov2012}]
    Given a pair of states $(x_0$, $x_d)$, the pair is called $k$ fixed-time controllable if there exists a sequence of control  $u(0)$,...,$u(k-1)$ that steers the BCN~\eqref{eq:1} from $x(0) = x_0$ to $x(k) = x_d$.
    The BCN~\eqref{eq:1} is $k$ fixed-time controllable if all pairs $(x_0,x_d)$ are $k$ fixed-time controllable.
\end{defnn}

For each ordered pair of states ($\delta_N^i, \delta_N^j$), we define two sets $\rho(i,j)$ and $\sigma(i,j)$ as follows: for each positive integer $k$, if there is  a sequence of control   $u(0)$,...,$u(k-1)$ that steers the BCN from $x(0) = \delta_N^i$ to $x(k) = \delta_N^j$, then $k\in\rho(i,j)$; otherwise, $k\in \sigma(i,j)$. It should be clear that $\rho(i,j)\sqcup \sigma(i,j) = \mathbb{Z}^+$. As a reference, we call $\rho(i,j)$ the \emph{set of reachable time steps} and $\sigma(i,j)$ the \emph{set of unreachable time steps}.

%the integer $k \in \rho(i,j)$ (resp.~$\sigma(i,j)$), then there is  a (resp. no any) sequence of control   $u(0)$,...,$u(k-1)$ such that steers the BCN from $x(0) = \delta_N^i$ to $x(k) = \delta_N^j$. Particularly, it holds that $\rho(i,j)\cap \sigma(i,j)=\emptyset$ and $\rho(i,j)\cup \sigma(i,j)=\mathbb{N}$.
With the above definitions, we present the categorization problem as follows.
\begin{problem}\label{prob}
    Consider the BCN~\eqref{eq:1}.  %with the reachable and unreachable time-step sets, i.e.,  $\rho(i,j)$ and $\sigma(i,j)$, respectively.
    The goal is to classify all pairs ($\delta_N^i, \delta_N^j$) into the four categories:
    \begin{enumerate}
        \item {\rm unreachable}: $|\rho(i,j)|=0$;% reachable
        \item {\rm transient}: $0<|\rho(i,j)|<\infty$ and $|\sigma(i,j)|=\infty$; % Transient category  % but $r_{ij}$ has infinite
        \item {\rm primitive}: $|\rho(i,j)|=\infty$ and $|\sigma(i,j)|<\infty$; % Primitive category
        \item {\rm imprimitive}: $|\rho(i,j)|=\infty$ and $|\sigma(i,j)|=\infty$. % Reducible category
    \end{enumerate}
\end{problem}

Equivalently, one wishes to obtain the {\rm controllability categorization matrix}
    %\begin{equation}
        $\mathcal{C}=(\mathcal{C}_{ji}),~ \mathcal{C}_{ji} \in[0, 3]$,
    %\end{equation}
    where $\mathcal{C}_{ji}$ is defined to be $k$ if the pair $(\delta_N^i,~\delta_N^j)$ belongs to the category $(k+1)$.

\subsection{Backgrounds}

We note that, as the number of state pairs in Problem~\ref{prob} is huge, and the BCN~\eqref{eq:1} can have complicated structures, solving Problem~\ref{prob} requires nontrivial techniques. We will develop a method via algebraic graph theoretic approach. Prior to that, we introduce in this subsection some preliminary results on digraphs and matrices, as well as the results on controllability of BCNs.

\subsubsection{Directed graphs}
%We first present some preliminaries about directed graphs (digraphs).

    Let $G = (V, E)$ be a digraph with the set of nodes (vertices) $V$ and the set of directed edges $E \subseteq V\times V$. The \emph{order} of a graph $G$ is the number of nodes in $V$. We denote by $v_i \rightarrow v_j$ a directed edge from $v_i$ to $v_j$ in $G$, and if $i=j$, the edge is called the \textit{self-loop} of the node $i$. The \textit{adjacency matrix} $A  \in {\mathscr{B}}_{n\times n}$ of $G$ is defined as follows: $A_{ij}=1$ (resp. $0$)  if and only if $v_i \rightarrow v_j \in (resp. \notin) ~E$.    %In other words, $A$ completely characterizes  $G$ and
    For simplicity, the digraph (i.e., $G$) of $A$ is denoted by $\mathcal{G}(A)$.

    Assumed that $v_i$ and $v_j$ are two nodes of $G$. A \textit{walk} from $v_i$ to $v_j$, denoted by $w_{ij}$, is a sequence of nodes $v_{i_0}\rightarrow v_{i_1}\rightarrow\cdots \rightarrow v_{i_m}$ in which each $v_{i_j}\rightarrow v_{i_{j+1}}$, for $j = 0,\ldots, m-1$, is an edge. If $v_{i_0}=v_{i_m}$, the walk is called a \emph{closed walk}. A \emph{cycle} is a closed walk with no repetition of nodes other than the starting- and the ending- node. A walk is said to be a {\em path} if all the nodes in the walk are pairwise distinct. Let $p_{ij}$ be a path from $v_i$ to $v_j$. We denote by $P_{ij}$ the set of all paths from $v_i$ to $v_j$. The length of a walk (resp. path, cycle) is the number of edges in that walk (resp. path, cycle).

    Two nodes  $v_i$ and $v_j$ of $G$ are called \textit{strongly connected} if there exists a directed walk from $v_i$ to $v_j$, and a directed walk from $v_j$ to $v_i$. A graph $G$ is strongly connected if any two nodes $v_i$ and $v_j$ are strongly connected. A single node with self-loop is regarded as trivially strongly connected to itself.  Evidently, strong connectivity between nodes is reflexive, symmetric, and transitive, resulting in an equivalence relation on the nodes of $G$ and simultaneously yielding a partition,
    %\begin{equation*}
        $V_1\sqcup V_2\sqcup\cdots\sqcup V_S$,
    %\end{equation*}
    with $\bigcup V_i=V$. Let $E_i$ be the set of edges $v_{i_j}\rightarrow v_{i_k}$ such that $v_{i_j}, v_{i_k}\in V_i$. Then $G_i=(V_i, E_i), i \in [1,S]$ are the \emph{induced subgraphs} of $G$. We also call each induced subgraph a \textit{strongly connected component} (SCC) of $G$. Specifically, a single node without self-loop is an SCC by itself. In this paper, we call such a single node the Type 1 SCC (T1SCC), and all other SCCs the Type 2 SCC  (T2SCC).

    %The induced subgraphs $G_1=(V_1, E_1) , G_2=(V_2, E_2) , . . . , G_d=(V_d, E_d)$ formed by preserving the nodes in an equivalence class and the edges among them, are called the \textit{strongly connected components }(SCCs) of $G$.
    We next present the following definition on condensation digraphs.

        \begin{defnn}[Condensation digraph~\cite{brualdi1991combinatorial}] \label{Condensation}
        Let $G$ be a digraph and $A$ be its adjacency matrix. Assume that $G$ has $S$ SCCs: $G_1,\ldots,G_S$, where $G_i=(V_i,E_i)$. %, $\mathcal{G}(A^{1,1}), \mathcal{G}(A^{2,2}), ... ,\mathcal{G}(A^{S,S})$.
        Let $\mathcal{G}^*(A)$ be the {\rm condensation digraph} of $\mathcal{G}(A)$, and $A^*$ be the adjacency matrix of $\mathcal{G}^*(A)$. $\mathcal{G}^*(A)$ and  $A^*$ are constructed as follows:
        \begin{enumerate}
            \item The set of nodes of $\mathcal{G}^*(A)$ is obtained by identifying each SCC as a node,
            \item If there exists a directed edge in $\mathcal{G}(A)$ from a node in $V_i$ to a node in $V_j$, %with $i<j$ (i.e., $A^{i,j}\neq \textbf{0}$),
            then $A^*_{ij}=1$; otherwise, $A^*_{ij}=0$.
        \end{enumerate}
         The constructed condensation digraph $\mathcal{G}^*(A)$ has no closed directed walks.
    \end{defnn}
    We then define the primitivity of a digraph.
\begin{defnn} [Primitive digraph~\cite{brualdi1991combinatorial}]\label{df:imprimitivity}
       Let $G$ be a strongly connected digraph of order $n$. Let $\eta=\eta(G)$ be the greatest common divisor of the lengths of the cycles of $G$. The digraph $G$ is primitive if $\eta = 1$ and  imprimitive if $\eta>1$. The integer $\eta$ is called the {\rm index of imprimitivity} of $G$. The index of imprimitivity $\eta$ is also referred to as the  {\rm loop number} \cite{gao2018stability}.
    \end{defnn}
    With Definition~\ref{df:imprimitivity}, we have the following result.
    \begin{lem}[\cite{brualdi1991combinatorial}] \label{imprimitivity_property}
        Let $G$ be a strongly connected digraph of order $n$ with index of imprimitivity $\eta$. Then, for each pair of nodes $v_i$ and $v_j$, the lengths of the directed walks from $v_i$ to $v_j$ are congruent modulo $\eta$.
        %\begin{enumerate}
%            \item For each vertex $a$ of $G$, $k$ equals the greatest common divisor of the lengths of the closed directed walks containing $a$.
%            \item For each pair of vertices $a$ and $b$, the lengths of the directed walks from $a$ to $b$ are congruent modulo $k$.
%            \item The set $V$ of vertices of $G$ can be partitioned into $k$ nonempty sets $V_1,V_2,...,V_k$, called the {\rm sets of imprimitivity} of $G$, where, with $V_{k+1} = V_1$, each edge of $D$ issues from $V_i$ and enters $V_{i+1}$ for some $i$ with $1\leq i \leq k$.
%            \item For $x_i \in V_i$ and $x_j \in V_j$ the length of a directed walk from $x_i$ t o $x_j$ is congruent to $j-i$ modulo $k$. More precisely, there exists a positive integer $N$ such that there are directed walks from $x_i$ to $x_j$ of every length $j - i + tk$ with $t\geq N$, ($1 \leq i , j \leq k$).
%        \end{enumerate}
    \end{lem}

\subsubsection{Matrices}
We first define the reducibility of a matrix.
\begin{defnn}[Reducible matrix~\cite{brualdi1991combinatorial}]\label{df:reducible}
        A matrix $A$ of order $n$ is called reducible if there exists a permutation matrix $P \in {\mathcal{L}}_{n\times n}$ %, and an integer $k$ with $1\leq k \leq n-1$
        such that
         \begin{equation}\label{eq:redu}
            P^{\top}AP=
            \left(
              \begin{array}{cc}
                B & C \\
                \textbf{0} & D \\
              \end{array}
            \right)
         \end{equation}
         where  $B$ and $D$ are square matrices of order at least one. A matrix is said to be irreducible if it is not reducible. % $B \in \mathbb{R}^{k\times k}$, $D \in \mathbb{R}^{(n-k)\times (n-k)}$, $C  \in \mathbb{R}^{k\times (n-k)}$ and $\textbf{0} \in \mathbb{R}^{(n-k)\times k}$ is a zero matrix.
    \end{defnn}

%If a matrix $M$ has all entries being nonnegative, {\color{red} the \emph{Boolean form} of $M$, denoted by $\mathscr{B}(M)$, is a Boolean matrix with the $ij$th entrie $1$ if $M_{ij}>0$, and the $ij$th entrie $0$ if $M_{ij}=0$.}
For the rest of this paper, we let $A$ be a Boolean matrix, i.e., all entries of $A$ are either $0$ or $1$.  It should be clear that there is an one-to-one correspondence between the set of Boolean matrices of order $n$ and the set of digraphs of order $n$. We then have the following lemmas.
\begin{lem}[\cite{brualdi1991combinatorial}] \label{irreducible_strong}
        The matrix  $A$ of order $n$ is irreducible if and only if its digraph $\mathcal{G}(A)$ is strongly connected.
\end{lem}
\begin{lem}[\cite{brualdi1991combinatorial}] \label{lem:norm}
        Let $A$ be a matrix of order $n$. Then there exists a permutation matrix $P$ of order $n$ and an integer $S\geq 1$ such that the {\rm Frobenius normal form} of $(\ref{eq:redu})$ can be written as
        \begin{equation}\label{normal_form}
            P^{\top}AP= \left(
                \begin{array}{cccc}
                    A^{1,1} &A^{1,2} &\cdots &A^{1,S}\\
                    \textbf{0}  &A^{2,2} &\cdots &A^{2,S} \\
                    \vdots &\vdots  &\ddots & \vdots \\
                    \textbf{0} &\textbf{0}  &\cdots &A^{S,S}
                \end{array}
                \right).
        \end{equation}
        %where $A^{1,1},A^{2,2}, ..., A^{S,S}$ are either the square irreducible matrices (i.e., T2SCCs) or the 1-by-1 zero matrices (i.e., T1SCCs). %which are uniquely determined to within simultaneous permutation of their lines, but their ordering in is not necessarily unique.
        %The Eq.~$(\ref{normal_form})$ is the normal form of the Eq.~$(\ref{eq:redu})$.
\end{lem}
If $A$ in Lemma~\ref{lem:norm} is the adjacency matrix of digraph $G$, then $A^{1,1},A^{2,2},\ldots,A^{S,S}$ are adjacency matrices of the SCCs of $G$. Specifically, if the SCC $G_i$ is a T2SCC, then $A^{i,i}$ is a square irreducible matrix; if the SCC $G_i$ is a T1SCC, then $A^{i,i}$ is a 1-by-1 zero matrix.

We next define the primitivity of a matrix.
\begin{defnn}[Primitive matrix~\cite{brualdi1991combinatorial}]\label{df:matrixprimit}
    A nonnegative matrix $A\in \mathbb{R}^{n\times n}$ is primitive if there exits an integer $j\ge 1$ such that $A^j>0$. If $A$ is primitive, the smallest $j$ such that $A^j>0$ is called the {\rm exponent} of $A$, denoted by $\gamma(A)$.
\end{defnn}
We note that if matrix $A$ is primitive, then it is also irreducible. Further, we also have the following results on primitive matrices.
\begin{lem}[\cite{brualdi1991combinatorial}]\label{lem:gamma}
If A is primitive, then $\gamma(A)\leq (n-1)^2+1$.
\end{lem}
\begin{prop}[\cite{brualdi1991combinatorial}]\label{prop:biject}
A digraph $G$ is primitive if and only if its adjacency matrix $A$ is primitive.
\end{prop}

\subsubsection{Controllability and $k$ fixed-time controllability}

Given a BCN~\eqref{eq:1}, we can compute the matrices $\mathcal{M}$ and $\mathcal{F}$ as in~\eqref{controllabiiltymtraix}. The controllability and $k$ fixed-time controllability of the BCN can then be determined by the following theorems.
\begin{theorem}[\cite{aa18}]\label{thm:controlF}
    The BCN $(\ref{eq:1})$ is
    \begin{enumerate}
        \item controllable from $\delta_N^j$ to $\delta_N^i$, if and only if, $\mathcal{F}_{ij}=1$;
        \item controllable at $\delta_N^j$, if and only if, $Col_j(\mathcal{F})>0$;
        \item controllable, if and only if, $\mathcal{F}>0$.
    \end{enumerate}
\end{theorem}

%\begin{theorem}{\rm \cite{Laschov2012}} \label{las_1}
%    The BCN $(\ref{eq:1})$ is controllable if and only if $\mathcal{M}$ is irreducible.
%\end{theorem}

\begin{theorem}[\cite{Laschov2012}] \label{las_2}
    Consider the BCN $(\ref{eq:1})$.
    \begin{enumerate}
        \item The BCN $(\ref{eq:1})$ is controllable, if and only if, $\mathcal{M}$ is irreducible.
        \item The BCN $(\ref{eq:1})$ is $k$ fixed-time controllable, if and only if,  $\mathcal{M}^{(k)} > 0$.
        \item If the matrix $\mathcal{M}$ is primitive, then $\gamma(\mathcal{M})\leq N^2-2N+2$ and the BCN $(\ref{eq:1})$ is $\gamma(\mathcal{M})$ fixed-time controllable. If $\mathcal{M}$  is not primitive, then the BCN is not $k$ fixed-time controllable for any $k$.
        \item If the BCN $(\ref{eq:1})$ is $k$ fixed-time controllable, then it is $p$ fixed-time controllable for any $p\geq k$.
    \end{enumerate}
\end{theorem}
\begin{rmk}\label{rem:01} %In \cite{Cheng2009controllability}
    We note that the controllability of the BCN can be determined with matrix $\mathcal{F}$ by Theorem~\ref{thm:controlF}, while whether the BCN is the $k$ fixed-time controllable or not can only be determined with matrix $\mathcal{M}$ by Theorem~\ref{las_2}. Specifically, if $\mathcal{M}$ is primitive, then BCN $(\ref{eq:1})$  is $k$ fixed-time controllable for any $k \geq \gamma(\mathcal{M})$. When $\mathcal{M}$ is reducible, although the BCN~\eqref{eq:1} is not $k$ fixed-time controllable, we may still have some pairs $(x_0,x_d)$ that are $k$ fixed-time controllable.

    %Considering the controllability matrix $\mathcal{F}$ in Eq.~$(\ref{controllabiiltymtraix})$, the time-step information (i.e., the superscript $i$ of $\mathcal{M}$) is integrated into $\mathcal{F}$, whereas it can not be recovered from $\mathcal{F}$, signifying hopeless on the detailed controllability analysis of the time-step information and even the  long-term controllability behavior. Instead, Laschov and Margaliot \cite{Laschov2012} studied two kinds of controllability based on the properties (irreducible and primitive) of the nonnegative matrix $\mathcal{M}$ in the Theorem~$\ref{las_2}$. Especially, when $\mathcal{M}$ is primitive, then the BCN $(\ref{eq:1})$  is k fixed-time controllable for any $k \geq \gamma(\mathcal{M})$, which reveals the long-term controllability behavior.   Particularly,  a controllable BCN is either P-controllable or NP-controllable according as  $\mathcal{M}$ being primitive and imprimitive.
\end{rmk}
For convenience, we call BCN~\eqref{eq:1} \textit{P-controllable}  if it is $k$ fixed-time controllable for some integer $k>0$. We call BCN~\eqref{eq:1} \textit{NP-controllable} if it is controllable, but not $k$ fixed-time controllable for any $k>0$. Equivalently, BCN~\eqref{eq:1} is P-controllable if $\mathcal{M}$ is primitive; BCN~\eqref{eq:1} is NP-controllable if $\mathcal{M}$ is irreducible and imprimitive.

%textit{NP-controllable},  if BCN $(\ref{eq:1})$ is $k$ fixed-time controllable for some integer $k>0$ (i.e., $\mathcal{M}$ is primitive) or controllable but not $k$ fixed-time controllable for any integer $k>0$ (i.e., $\mathcal{M}$ is imprimitive), respectively.

%\subsection{Controllability and $k$ fixed-time controllability}

\section{Main Results}\label{section3}

Recall that in Problem~\ref{prob}, we aim to classify all state pairs of BCN~\eqref{eq:1} into four categories.
Equivalently, one wishes to obtain the {\rm controllability categorization matrix}
        $\mathcal{C}$.  We note that the Boolean form of $\mathcal{C}$ is exactly $\mathcal{F}$ in~\eqref{controllabiiltymtraix}, i.e.,
        $\mathscr{B}(\mathcal{C})\equiv\mathcal{F}$.

Evidently, the form of the controllability categorization matrix $\mathcal{C}$ is trivial in the following situation.
    %\begin{enumerate}
      If the BCN $(\ref{eq:1})$ is P-controllable (resp. NP-controllable), then, $\mathcal{C}=\textbf{2}_{N\times N}$ (resp. $\mathcal{C}=\textbf{3}_{N\times N}$).
    %\end{enumerate}
    In other words, $\mathcal{C}$ is trivial if $\mathcal{M}$ is irreducible, as it follows from the definitions that
    \begin{enumerate}
    \item If the BCN $(\ref{eq:1})$ is P-controllable, then for any pair of states $\delta_N^i$ and $\delta_N^j$, we have $|\rho(i,j)|=\infty$ and $|\sigma(i,j)|<\infty$.
    \item If the BCN $(\ref{eq:1})$ is NP-controllable, then for any pair of states $\delta_N^i$ and $\delta_N^j$, we have $|\rho(i,j)|=\infty$ and $|\sigma(i,j)|=\infty$.
    \end{enumerate}

    In the rest of this section, we investigate the case when $\mathcal{M}$ is reducible.

\subsection{Main theorem}

Let $G=(V,E)$ be the \textit{state transition digraph} of BCN~\eqref{eq:1}, where $V$ is the set of states, i.e., $V:=\Delta_N$, and $E:=\{\delta_N^i\rightarrow \delta_N^j~|~ \delta_N^j=Lu\delta_N^i \mbox{~for~some~} u \in \Delta_M\}$, i.e., an edge $\delta_N^i\rightarrow\delta_N^j$ exists in $E$ if there exists some control $u$ which drives the system from state $\delta_N^i$ to state $\delta_N^j$ in one step. Let $\bar{\mathcal{M}}$ be the adjacency matrix of $G$, then, $\bar{\mathcal{M}}:=\mathcal{M}^{\top}$, where $\mathcal{M}$ is defined as in~\eqref{controllabiiltymtraix}. Then, by Lemma~\ref{lem:norm}, we can write $\bar{\mathcal{M}}$ in the Frobenius normal form as in~\eqref{normal_form} in a similar manner, with the replacement of $A$ in~\eqref{normal_form} with $\bar{\mathcal{M}}$. Then we have that
$\bar{\mathcal{M}}^{1,1}\in \mathscr{B}^{q_1\times q_1},\bar{\mathcal{M}}^{2,2} \in\mathscr{B}^{q_2\times q_2}, ..., \bar{\mathcal{M}}^{S,S}\in \mathscr{B}^{q_S\times q_S}$ where $\sum_{i=1}^S q_i=N$.
The BCN~\eqref{eq:1} thereby has $S$ SCCs, denoted by, $\mathcal{X}_1,..., \mathcal{X}_S$,
%\begin{equation*}
%    \mathcal{X}_1=\{x_1^1, ...,  x_1^{q_{1}}\},..., \mathcal{X}_S=\{x_S^1, ...,  x_S^{q_{S}}\},
%\end{equation*}
with $\bigcup_{i=1}^S\mathcal{X}_i=\Delta_N$. Additionally, we denote by $Id(i)$ the index of the SCC that the state $\delta_N^i$ belongs to. From the Definition~\ref{Condensation}, we can construct the condensation digraph $\mathcal{G}^*(\bar{\mathcal{M}})$ of the BCN~\eqref{eq:1}  as well as its adjacency matrix $\bar{\mathcal{M}}^*$.

Given a pair of states $\delta_N^i$ and $\delta_N^j$, which are two nodes in the state transition graph $G$, let $P_{ij}$ be the set of paths from $\delta_N^i$ to $\delta_N^j$, and $P^*_{Id(i),Id(j)}$ be the set of paths from $Id(i)$ to $Id(j)$ in the condensation graph $\mathcal{G}^*(\bar{M})$. Let $K:=|P^*_{Id(i),Id(j)}|$. We denote these paths in the condensation graph by $p^{*1}_{Id(i),Id(j)},p^{*2}_{Id(i),Id(j)},\ldots,p^{*K}_{Id(i),Id(j)}$. Then, we use the following method to partition the set $P_{ij}$ into $K$ subsets $P_{ij}^1,\ldots, P_{ij}^K$.

\emph{For any path $p_{ij}\in P_{ij}$, we replace every node $\delta_N^l\in p_{ij}$ with the node $Id(l)$, if the resulting path, ignoring self-loops, is $p^{*k}_{Id(i),Id(j)}$, then $p_{ij}\in P_{ij}^k$.}

With the above partitions, we further define $\eta_{ij}^k$ to be the greatest common divisor of the indexes of primitivity of the T2SCC along path $p^{*k}_{Id(i),Id(j)}$. If there is no T2SCC along the path  $p^{*k}_{Id(i),Id(j)}$, we let $\eta_{ij}^k=0$ and $\pi_{ij}^k=\emptyset$. Otherwise, we define
\begin{equation}\label{pi_ijk}
        \pi_{ij}^k = \left\{l(p_{ij}^k) \mod \eta_{ij}^k~|~p_{ij}^k\in P_{ij}^k\right\},
\end{equation}
let $\bar{\eta}_{ij}$ be the least common multiple of $\{\eta_{ij}^1,\ldots,\eta_{ij}^K\}$, and  $\Tilde{\pi}_{ij}^k:=\left\{a + b\cdot  \eta_{ij}^k~|~a \in \pi_{ij}^k,~b \in \left[0, \frac{\bar{\eta}_{ij}}{\eta_{ij}^k} - 1\right]\right\}$. Then, let $\bar{\pi}_{ij}:=\bigcup \Tilde{\pi}_{ij}^k$. With these definitions, we are in a position to present our main theorem.

    \begin{theorem}\label{thm:main02}
        Considering the BCN $(\ref{eq:1})$,   we have
        \begin{enumerate}
            \item $\mathcal{C}_{ji}=0$, if and only if, $\mathcal{F}_{ji}=0$.
            \item $\mathcal{C}_{ji}=1$, if and only if, $\bar{\pi}_{ij}=\emptyset$ and $\mathcal{F}_{ji}=1$.
            \item $\mathcal{C}_{ji}=2$, if~and~only~if, $\bar{\pi}_{ij}=[0, \bar{\eta}_{ij}-1]$.
            \item $\mathcal{C}_{ji}=3$, if and only if, $\bar{\pi}_{ij} \neq \emptyset$ and $\bar{\pi}_{ij}\neq [0, \bar{\eta}_{ij}-1]$.
        \end{enumerate}
    \end{theorem}
    We first provide a proof for the bulletins~(1) and~(2) of Theorem~\ref{thm:main02}.
In the next subsection, we will provide a complete proof of bulletins~(3) and~(4). We will also provide a follow-up result in section~\ref{last}.

\begin{proof}[Proof of Theorem~\ref{thm:main02}, part~I]
	We now prove the first two bulletins of Theorem~\ref{thm:main02}.
	\begin{enumerate}
		\item We note that by definition, $\mathcal{C}_{ji}=0$ if and only if $\delta_N^j$ is unreachable from $\delta_N^i$, or equivalently, BCN~\eqref{eq:1} is uncontrollable from $\delta_N^i$ to $\delta_N^j$, which, by Theorem~\ref{thm:controlF}, holds if and only if $\mathcal{F}_{ji}=0$.
		\item We show that if $\bar{\pi}_{ij}=\emptyset$ and $\mathcal{F}_{ji}=1$, then $\mathcal{C}_{ji}=1$. The other direction can be similarly shown. By definition, $\mathcal{C}_{ji}=1$ if and only if $0<|\rho(i,j)|<\infty$ and $|\sigma(i,j)|=\infty$. Note that $\bar{\pi}_{ij}=\emptyset$ if and only if $\Tilde{\pi}_{ij}^k=\emptyset$ (i.e., $\pi_{ij}^k=\emptyset$) for all $k \in [1,K]$. This implies that $\eta_{ij}^k=0$ for all $k$ and there is no T2SCC along the path $p^{*k}_{Id(i),Id(j)}$ for all $k$. Therefore, there is no T2SCC along any path $p_{ij}\in P_{ij}$. We thus have that $\rho(i,j) = \{l(p_{ij})\mid p_{ij}\in P_{ij}\}$. Since $|P_{ij}|$ is finite, we have that $|\rho(i,j)|<\infty$. Since $\mathcal{F}_{ji}=1$, there exists at least one path $p_{ij}\in P_{ij}$, which implies that $|P_{ij}|>0$ and $|\rho(i,j)|>0$. Lastly, for any positive integer $k\notin \{l(p_{ij})\mid p_{ij}\in P_{ij}\}$, we have that $k\in \sigma(i,j)$. This implies that $|\sigma(i,j)|=\infty$.
	\end{enumerate}
\end{proof}

\subsection{Analysis and proof of Theorem~\ref{thm:main02}}

In this subsection, we show the last two bulletins of Theorem~\ref{thm:main02}.  We first consider the case that $|P^*_{Id(i), Id(j)}|=1$, i.e., there is only one path from $Id(i)$ to $Id(j)$ in the condensation graph. Later we will extend to the general case where $|P^*_{Id(i), Id(j)}|=K$ for any positive integer $K$.

For ease of notation, we now use $p_{Id(i),Id(j)}^*$ to denote the path from $Id(i)$ to $Id(j)$ in the condensation graph, and let $\eta_{ij}$ be the greatest common divisor of the indexes of imprimitivity of the T2SCCs along the path $p_{Id(i),Id(j)}^*$. Again, if there is no T2SCC along the path, let $\eta_{ij}=0$.

Let the cycles in the T2SCCs along path $p_{Id(i),Id(j)}^*$ be $c_1,...,c_k$ with the lengths equal to $l(c_1),...,l(c_k)$, respectively. Then any walk $w_{ij}$ of  $\mathcal{G}(\bar{\mathcal{M}})$ has length of the form
    \begin{equation*}
        l(w_{ij})=l(p_{ij}) + a_1 \cdot l(c_1)+\cdots+a_k \cdot l(c_k),
    \end{equation*}
    where $a_1,\ldots, a_k$ are nonnegative integers.
    Note that, from~\cite{brualdi1991combinatorial}, we have the following lemma.

    \begin{lem} [\cite{brualdi1991combinatorial}]\label{Frobenius_Schur}
        Let $\Psi=\{l_1, l_2,...,l_k\}$ be a nonempty set of positive integers and  $\eta$ be the greatest common divisor of the integers in $\Psi$. Then there  exists a smallest positive integer $\phi(l_1, l_2,...,l_k)$, called the {\rm Frobenius-Schur index} of $\Psi$, such that for any integer $n\geq \phi(l_1, l_2,...,l_k)$, $n\eta$ can be expressed as a nonnegative linear combination of  these integers, i.e., as a sum,
        %\begin{equation*}
            $n\eta=a_1l_1+a_2l_2+\cdots+a_kl_k$,
        %\end{equation*}
        where $a_1,a_2,...,a_k$ are nonnegative integers.
    \end{lem}

    From the Lemma~\ref{Frobenius_Schur} and the definition of $\eta_{ij}$, there exists a $\phi\left(l(c_1),...,l(c_k)\right)$ such that for any $n\ge \phi\left(l(c_1),...,l(c_k)\right)$, we have that $n\eta_{ij} = a_1 l(c_1)+\cdots+a_k l(c_k)$. Similarly, we can find some $\phi'\left(l(c_1),...,l(c_k)\right)\ge \phi\left(l(c_1),...,l(c_k)\right)$ such that for any integer $n\ge\phi'\left(l(c_1),...,l(c_k)\right)$, we have that $n\eta_{ij} = a_1' l(c_1)+\cdots+a_k' l(c_k)$ where $a_1',\ldots,a_k'$ are positive integers. Therefore, by connecting these cycles with a path $p_{ij}$, we can obtain a walk $w_{ij}$ whose length can be $l(p_{ij})+n\eta_{ij}$ for all $n\ge \phi'\left(l(c_1),...,l(c_k)\right)$.

    As in~\eqref{pi_ijk}, we define a set of integers,
    \begin{equation}\label{eq:pi_ij}
        \pi_{ij} = \left\{l(p_{ij}) \mod \eta_{ij}~|~p_{ij}\in P_{ij}\right\}.
    \end{equation}
    Evidently, if $\pi_{ij}=[0, \eta_{ij}-1]$, i.e., $\pi_{ij}$ is a complete residue system modulo $\eta_{ij}$,  then  for any integer $n\geq \phi'\left(l(c_1),...,l(c_k)\right)\eta_{ij}+\max_{p_{ij}\in P_{ij}}l(p_{ij})$, there exists a walk $w_{ij}$ of length $n$.
    In the special case that $\eta_{ij}=0$, then $\pi_{ij}$ is not well defined. We thereby redefine $\pi_{ij}=\emptyset$ for such a case.

    Based on the above definitions, we have the following result.
    \begin{lem}\label{lem:main01}
        Let $\delta_N^i$ and $\delta_N^j$ be two nodes of the state transition digraph $\mathcal{G}(\bar{\mathcal{M}})$ of the BCN $(\ref{eq:1})$.  Consider the case that  there is only one path $p_{Id(i),Id(j)}^*$ from node $Id(i)$ to node $Id(j)$ in the condensation digraph $\mathcal{G}^*(\bar{\mathcal{M}})$. Then the pair $(\delta_N^i, \delta_N^j)$ belongs to the primitive category, i.e., $\mathcal{C}_{ji}=2$, if and only if $\pi_{ij}=[0,\eta_{ij}-1]$.
    \end{lem}

\begin{proof}
	({\rm Sufficiency}). When $\pi_{ij}=[0,\eta_{ij}-1]$, by the arguments before Lemma~\ref{lem:main01}, there exists an integer
	\begin{align}\label{nprime}
	n':=\phi'\left(l(c_1),...,l(c_k)\right)\eta_{ij}+\max_{p_{ij}\in P_{ij}}l(p_{ij})>0
	\end{align}
	such that one can construct a walk $w_{ij}$ with the length $n$, for any integer $n\geq n'$. In other words,  $|\rho(i,j)|=\infty$ and $|\sigma(i,j)|<\infty$. This implies $\mathcal{C}_{ji}=2$.
	
	({\rm Necessity}). When $\mathcal{C}_{ji}=2$, we have that $|\rho(i,j)|=\infty$ and $|\sigma(i,j)|<\infty$. Suppose that to the contrary $\pi_{ij}\ne[0,\eta_{ij}-1]$. Then there exists an integer $k\in [0,\eta_{ij}-1]$ such that $k \notin \pi_{ij}$, which implies that a walk of length $\eta_{ij}n+k$, $\forall n \ge n'$, cannot be constructed. This contradicts with $|\sigma(i,j)|<\infty$.
\end{proof}

Note that Lemma~\ref{lem:main01} essentially proves the third bulletin of Theorem~\ref{thm:main02} for the case when $|P^*_{Id(i), Id(j)}|=1$. We next consider the general case where $|P^*_{Id(i), Id(j)}|=K$, with $K$ being any positive integer, and prove the last two bulletins of Theorem~\ref{thm:main02}.

\begin{proof}[Proof of Theorem~\ref{thm:main02}, part~II]
	\begin{enumerate}
		\item[(3)]
		(Sufficiency). Recall the definitions before Theorem~\ref{thm:main02}. Each $\eta_{ij}^k$ is the greatest common divisor of the index of primitivities of the T2SCCs along path $p_{Id(i),Id(j)}^{*k}$. Similar to the arguments for the case $|P^*_{Id(i), Id(j)}|=1$, it can be shown that if $s\in \pi_{ij}^k$, we can construct a walk $w_{ij}$ of length $n_k\eta_{ij}^k + s$ for any $n_k\ge n_k'$ for some positive integer $n_k'$. Here, we can pick $n_k'$ as in~\eqref{nprime}. We perform the same implementation for each path $p_{Id(i),Id(j)}^{*k}$. Then, with the definition of $\bar{\eta}_{ij}$, for some $n_k''$, we can rewrite the set
		$\{n_k\eta_{ij}^k + s\mid n_k\ge n_k'\}$ as $\{n_k\bar{\eta}_{ij} + s, n_k\bar{\eta}_{ij} + \eta_{ij}^k + s, ...,  n_k\bar{\eta}_{ij} + \eta_{ij}^k \frac{\bar{\eta}_{ij}}{\eta_{ij}^k} + s\mid n_k\ge n_k''\}$. Note that this can be done for each $k\in[1, K]$. Therefore, for each $s^*\in \bar{\pi}_{ij}$, we have that $s^*\in\Tilde{\pi}_{ij}^k$ for some $k$, and a walk of length $n_k\bar{\eta}_{ij}+s^*$, $n_k\ge n_k''$, can be constructed. If $\bar{\pi}_{ij} = [0,\bar{\eta}_{ij}-1]$, then there exists some $n^*=\max_{k\in [1,K]} n_k''$ such that for any $n\ge n^*$, we can construct a walk of length $n$. This implies that $|\rho(i,j)|=\infty$ and $|\sigma(i,j)|<\infty$, or equivalently, $\mathcal{C}_{ji}=2$.
		
		\noindent
		(Necessity). Suppose that to the contrary $\bar{\pi}_{ij} \ne [0,\bar{\eta}_{ij}-1]$, then there exists some $s\in [0,\bar{\eta}_{ij}-1] $ such that $s\notin \Tilde{\pi}_{ij}^k$ for any $k \in [1,K]$, which implies that we cannot construct a walk of length $n\bar{\eta}_{ij}+s,\forall n\in\mathbb{Z}^+$. This implies that $|\sigma(i,j)|=\infty$, which is a contradiction.
		\item[(4)] Since we have shown~(1), (2), (3), the result of~(4) follows directly.
	\end{enumerate}
\end{proof}

\subsection{Connection of categorization results to graph structure}\label{last}
In this subsection, we provide a further result on the categorization of state pairs. In particular, we show the following fact which relates the categorization to the structure of the state transition digraph.

\begin{theorem}\label{lem:Condensationmatrix}
    Let $(\delta_N^i,\delta_N^j)$ be a pair of states of BCN~\eqref{eq:1}. Suppose that $\delta_N^i\in \mathcal{X}_\alpha$ and $\delta_N^j\in\mathcal{X}_\beta$, where $\mathcal{X}_{\alpha}$ and $\mathcal{X}_{\beta}$ are two SCCs of the state  transition digraph. Then, for any states $\delta_N^{i'}\in\mathcal{X}_\alpha$ and $\delta_N^{j'}\in\mathcal{X}_{\beta}$, we have that $\mathcal{C}_{j'i'}=\mathcal{C}_{ji}$.
        %Let $\mathcal{X}_{\alpha}$ and $\mathcal{X}_{\beta}$ be two SCCs of the state  transition  digraph of BCN~\eqref{eq:1}.
        %Then, for all pair $(\delta_N^i,\delta_N^j)$ with $\delta_N^i \in \mathcal{X}_{\alpha}$ and $\delta_N^j \in \mathcal{X}_{\beta}$,  we have that all $\mathcal{C}_{ji}$ preserve the unified value, which is in the integer set $[0,3]$.
\end{theorem}
    We now prove the above theorem. To proceed, we first recall some notations.
    Let $\delta_N^i$ and $\delta_N^j$ be two nodes of the state transition digraph $\mathcal{G}(\bar{\mathcal{M}})$ of the BCN $(\ref{eq:1})$. Denote by $\alpha := Id(i)$ and $\beta = Id(j)$. Let $\eta_{\alpha}$ (resp. $\eta_{\beta}$) be the index of imprimitivity of the SCC $\mathcal{X}_{\alpha}$ (resp. $\mathcal{X}_{\beta}$). In the special case that $\mathcal{X}_{\alpha}$ (resp. $\mathcal{X}_{\beta}$) is a T1SCC, we redefine $\eta_{\alpha}:=0$ (resp.  $\eta_{\beta}:=0$).

    With the above notations, the proof of Theorem~\ref{lem:Condensationmatrix} is shown as follows.

\begin{proof}[Proof of Theorem~\ref{lem:Condensationmatrix}]
	It should be clear that if there is no path from any node in $\mathcal{X}_{\alpha}$  to any node in $\mathcal{X}_{\beta}$, then, we have that $\mathcal{C}_{ji}=0$, $\forall \delta_N^i \in \mathcal{X}_{\alpha}, \forall \delta_N^j \in \mathcal{X}_{\beta}$.  We now restrict our discussion to the situation that there exists a path from some node in $\mathcal{X}_{\alpha}$ to some node in $\mathcal{X}_{\beta}$.
	
	First, consider the case that $\alpha=\beta$. If there is only one node in $\mathcal{X}_{\alpha}$ (i.e., T1SCC or T2SCC), then the theorem trivially holds. If $\mathcal{X}_{\alpha}$ is a T2SCC with at least two nodes, one can prove the theorem as follows. (1) If $\eta_{\alpha}=1$, then we have that $\mathcal{C}_{ji}=2, \forall \delta_N^i, \delta_N^j \in \mathcal{X}_{\alpha}$. (2) If $\eta_{\alpha}>1$, then by the Lemma \ref{imprimitivity_property}, for any pair $\delta_N^i, \delta_N^j \in \mathcal{X}_{\alpha}$, there exists an integer $k\in[0, \eta_{\alpha}]$ such that $\pi_{ij}=\{k\}$. In other words, $\pi_{ij}\neq \emptyset$ and $\pi_{ij}\neq [0, \eta_{\alpha}]$. By Theorem \ref{thm:main02}, this implies that  $\mathcal{C}_{ji}=3, \forall \delta_N^i, \delta_N^j \in \mathcal{X}_{\alpha}$.
	
	Next, consider the case that $\alpha\neq \beta$. For any pair ($\delta_N^i,\delta_N^j$) with $\delta_N^i  \in \mathcal{X}_{\alpha}$ and $\delta_N^j \in \mathcal{X}_{\beta}$, we have that $\bar{\eta}_{ij}=\eta_{\alpha\beta}^*$ for some $\eta_{\alpha\beta}^*$.
	\begin{enumerate}
		\item If $\eta_{\alpha\beta}^*=1$, then,  $\bar{\pi}_{ij}=\{0\}=[0, \eta_{\alpha\beta}^*-1]$. This implies that $\mathcal{C}_{ji}=2, \forall \delta_N^i\in \mathcal{X}_{\alpha}, \delta_N^j \in \mathcal{X}_\beta$.
		\item If $\eta_{\alpha\beta}^*>1$, then for any $\delta_N^{j'}\in \mathcal{X}_\beta$, from Lemma~\ref{imprimitivity_property}, one can conclude that there exists some $s_1$ such that
		$$
		\bar{\pi}_{ij} + s_1 := \{s_1 + y \mod  \bar{\eta}_{ij}~|~ y \in \bar{\pi}_{ij}\} \equiv \bar{\pi}_{ij'}.
		$$
		For any $\delta_N^{i'}\in \mathcal{X}_\alpha$, we can also conclude that there exists some $s_2$ such that
		$$
		\bar{\pi}_{ij} + s_2 := \{s_2 + y \mod  \bar{\eta}_{ij}~|~ y \in \bar{\pi}_{ij}\} \equiv \bar{\pi}_{i'j}.
		$$
		Then we have that
		$$
		\bar{\pi}_{ij} + s: = \{s + y \mod  \bar{\eta}_{ij}~|~ y \in \pi_{ij}\} \equiv \bar{\pi}_{i'j'},
		$$
		where $s:=s_1-s_2$. This implies that $|\bar{\pi}_{i'j'}|=|\bar{\pi}_{ij}|$, $\forall \delta_N^{i'} \in \mathcal{X}_{\alpha}, \forall \delta_N^{j'} \in \mathcal{X}_{\beta}$. Then by Theorem~\ref{thm:main02}, we conclude that $\mathcal{C}_{j'i'}=\mathcal{C}_{ji}$.
	\end{enumerate}
\end{proof}

    Based on the Theorem~\ref{lem:Condensationmatrix}, we can define an $S\times S$ matrix $\mathscr{C}=(\mathscr{C}_{\beta\alpha}),~\mathscr{C}_{\beta\alpha}\in [0,3]$, where  $\mathscr{C}_{\beta\alpha}:=\mathcal{C}_{ji}$ with  $\delta_N^i \in \mathcal{X}_{\alpha}$ and $\delta_N^j \in \mathcal{X}_{\beta}$. We call $\mathscr{C}$ the \textit{condensation controllability categorization matrix} of BCN~\eqref{eq:1}.  Then Theorem~\ref{lem:Condensationmatrix} has the following equivalent expression.

    \noindent
{\bf Theorem~\ref{lem:Condensationmatrix}} (An alternative version).
\emph{
        Let $\mathcal{X}_{\alpha}$ and $\mathcal{X}_{\beta}$ be two SCCs of the state  transition  digraph of the BCN $(\ref{eq:1})$.
        Then, the pair $(\delta_N^i,\delta_N^j)$ with $\delta_N^i \in \mathcal{X}_{\alpha}$ and $\delta_N^j \in \mathcal{X}_{\beta}$ belongs to the
        \begin{enumerate}
            \item unreachable category, if and only if,  $\mathscr{C}_{\alpha,\beta}=0$;
            \item transient category, if and only if,  $\mathscr{C}_{\alpha,\beta}=1$;
            \item primitive category, if and only if,  $\mathscr{C}_{\alpha,\beta}=2$;
            \item imprimitive category, if and only if,  $\mathscr{C}_{\alpha,\beta}=3$.
        \end{enumerate}
}

   \begin{rmk}\label{rem:Condensationmatrix}
        We note that our definition of condensation controllability categorization matrix $\mathscr{C}$ is a generalization of the so-called {\rm reduced controllability matrix}  $\mathscr{B}(\mathscr{C})$ in~\cite{zhu2019further}. Notably, the dimension of $\mathscr{C}$ may be much smaller than the one of $\mathcal{C}$, if the number of nodes of the condensation digraph of the  BCN~$(\ref{eq:1})$ is much smaller than the number of states of the  BCN~$(\ref{eq:1})$, i.e.,  $S\ll N$. In other words, to save the controllability information of BCNs, $\mathscr{C}$ is much better and more economical than $\mathcal{C}$.
    \end{rmk}

\section{Example}\label{section4}
In this section, we provide an example BCN as in~\cite{zhu2019further} to illustrate our results.  In particular, the algebraic form of the BCN is
\begin{equation}\label{ex1}
            x(t+1)=\delta_8[2,5,3,5,6,4,8,7,4,5,4,5,6,7,8,7]u(t)x(t),
    %x(t+1)=[\delta_8^2~\delta_8^5~\delta_8^3~\delta_8^5~\delta_8^6~\delta_8^4~\delta_8^8~\delta_8^7~
            %\delta_8^4~\delta_8^5~\delta_8^4~\delta_8^5~\delta_8^6~\delta_8^7~\delta_8^8~\delta_8^7]u(t)x(t),
\end{equation}
where $x(t)\in\Delta_8$, $u(t)\in\Delta_2$. The state transition digraph of the BCN~$(\ref{ex1})$ and its condensation digraph  are shown in the Fig.~\ref{fg:ex1A}-\ref{fg:ex1B}. In particular, the state transition digraph has $5$ SCCs, $\mathcal{X}_1=\{\delta_8^1\}$, $\mathcal{X}_2=\{\delta_8^2\}$, $\mathcal{X}_3=\{\delta_8^3\}$, $\mathcal{X}_4=\{\delta_8^4,\delta_8^5,\delta_8^6\}$, $\mathcal{X}_5=\{\delta_8^7,\delta_8^8\}$, and $Id(i)=i$, $i \in [1,3]$,  $Id(j)=4$, $j \in [4,6]$, $Id(k)=5$, $k \in [7,8]$. Notably, $\mathcal{X}_1$ and $\mathcal{X}_2$ are T1SCCs, whereas $\mathcal{X}_3$, $\mathcal{X}_4$ and $\mathcal{X}_5$ are T2SCCs with the indexes  of imprimitivity equal to 1, 3 and 2, respectively.

Here, we consider the pair ($\delta_8^1, \delta_8^4$). From the Fig.~\ref{fg:ex1A}, the path set $P_{14}$ from $\delta_8^1$ to $\delta_8^4$ has only two paths with the lengths 1 and 4, respectively.  In the the condensation digraph, i..e, Fig.~\ref{fg:ex1B}, there are also two paths from the node $Id(1)=1$ to node $Id(4)=4$. So the path set $P_{14}$ can be partitioned into two distinct sets, $P_{14}^1=\{p_{14}^1\}$ and $P_{14}^2=\{p_{14}^2\}$, where $p_{14}^1: \delta_8^1 \rightarrow \delta_8^4$ and $p_{14}^2: \delta_8^1 \rightarrow \delta_8^2 \rightarrow \delta_8^5 \rightarrow \delta_8^6 \rightarrow \delta_8^4$. Since there is only one cycle with the length 3 in the $\mathcal{X}_4$, we have $\eta_{14}^1=3$ and $\eta_{14}^2=3$. According to~\eqref{pi_ijk}, we have that $\Tilde{\pi}_{14}^1 = \pi_{14}^1 = \{1\}$ and $\Tilde{\pi}_{14}^2 = \pi_{14}^2 = \{2\}$. Hence,  $\bar{\eta}_{14}=3$ and $\bar{\pi}_{14} =\Tilde{\pi}_{14}^1 \sqcup \Tilde{\pi}_{14}^2=\{1,2\}\neq [0, 2]$. This, together with the Theorem~\ref{thm:main02}, implies that $\mathcal{C}_{41}=3$, i.e., the pair ($\delta_8^1, \delta_8^4$) belongs to the imprimitive category.

Akin to the analysis above,  the controllability categories of the other pairs can be obtained. Indeed, one can obtain the controllability categorization matrix. Then, based on the matrix $\mathcal{C}$, the condensation controllability categorization matrix can be induced. Both matrices are presented as follows,
\begin{equation*}
    {
    \mathcal{C}=\left(
            \begin{array}{c|c|c|ccc|cc}
                0 &0 &0 &0 &0 & 0& 0&0\\
                \hline
                1 &0 &0 &0 &0 & 0& 0&0\\
                \hline
                0 &0 &2 &0 &0 & 0& 0&0\\
                \hline
                3 &3 &1 &3 &3 & 3& 0&0\\
                3 &3 &1 &3 &3 & 3& 0&0\\
                3 &3 &1 &3 &3 & 3& 0&0\\
                \hline
                2 &2 &1 &2 &2 & 2& 3&3\\
                2 &2 &1 &2 &2 & 2& 3&3\\
            \end{array}
    \right),~\\
    \mathscr{C}=\left(
            \begin{aligned}
                0~ 0~ &0~ 0~ 0\\
                1~ 0~ &0~ 0~ 0\\
                0~ 0~ &2~ 0~ 0\\
                3~ 3~ &1~ 3~ 0\\
                2~ 2~ &1~ 2~ 3\\
            \end{aligned}
    \right).\\
    }
\end{equation*}
%Then, based on the the matrix $\bar{\mathcal{C}}:=\mathcal{C}^{\top}$, the condensation controllability category matrix is presented as follows,
%\begin{equation*}
%    \mathscr{C}=\left(
%            \begin{array}{ccccc}
%                0 &1 &0 &3 &2\\
%                0 &0 &0 &3 &2\\
%                0 &0 &2 &1 &1\\
%                0 &0 &0 &3 &2\\
%                0 &0 &0 &0 &3\\
%            \end{array}
%    \right).\\
%\end{equation*}

\begin{figure}[thb]%\label{fg1}
        \centering
        \begin{center}
            \includegraphics[width=8.6cm]{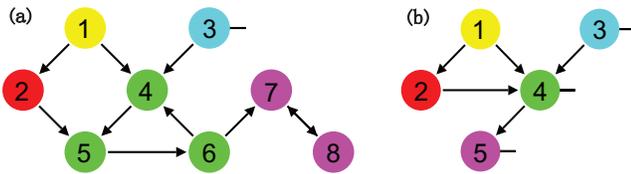}
            \subfigure{\label{fg:ex1A}}\subfigure{\label{fg:ex1B}}
        \end{center}
        \caption{This figure, originally from \cite{zhu2019further}, shows  (a) the state transition digraph  and (b) the condensation digraph of the BCN~\eqref{ex1}, respectively. For simplicity, in (a), the number $i$ in each circle denotes the state $\delta_8^i$, and in (b), the number $i$ in each circle represents the $i$th SCC $\mathcal{X}_i$.}\label{fg1}
        %\normalsize
\end{figure}

\section{Conclusions}\label{sec5}

In this paper, we have established a detailed analysis on the~$k$ fixed-time controllability of all state pairs of a BCN. The definition of the controllability categorozation matrix is first proposed, extending the conventional controllability matrix. By routinely using the algebraic form of BCNs and the algebraic digraph theory,  we have constructed the controllability categorization matrix. Then, a condensation controllability categorization matrix is also induced.  %Finally, we have provided an example to illustrate the usefulness of the analytical controllability framework that we have developed in the paper.
Overall, leveraging this framework may enable the development in the control-theoretic analysis of BCNs in the future. %such as, 1) design proper control law to stabilize the BCNs or synchronize a coupled BCNs, 2) extend the controllability category concept to some special kinds of BNs, such as CBNs, and 3) application to the real biological systems that are governed by the BCNs.
%\\
%\\
%\textbf{References}

% or
% do not use \section anymore after \appendix,  only \section*



%\section*{Acknowledgment}
%A toolbox for all the related computations is available
%at http://lsc.amss.ac.cn/dcheng/.

\ifCLASSOPTIONcaptionsoff
  \newpage
\fi


% can use a bibliography generated by BibTeX as a.bbl file
%\bibliographystyle{IEEEtran}
%\bibliography{IEEEabrv, ../bib/paper}
%
% <OR> manually copy in the resultant.bbl file
%\bibliographystyle{IEEEtran}
%\bibliography{bare_jrnl}

% Generated by IEEEtran.bst,  version: 1.13 (2008/09/30)

%\bibliographystyle{plain}
\bibliographystyle{ieeetr}
%\bibliographystyle{unsrt}
%\begin{thebibliography}{10}
\bibliography{ZQX}

\begin{thebibliography}{10}

\bibitem{aa1}
S.~Kauffman, ``Metabolic stability and epigenesis in randomly constructed
  genetic nets,'' {\em Journal of Theoretical Biology}, vol.~22, no.~3,
  pp.~{437--467}, 1969.

\bibitem{smolen2000mathematical}
P.~Smolen, D.~A. Baxter, and J.~H. Byrne, ``Mathematical modeling of gene
  networks,'' {\em Neuron}, vol.~26, no.~3, pp.~567--580, 2000.

\bibitem{albert2003topology}
R.~Albert and H.~Othmer, ``The topology of the regulatory interactions predicts
  the expression pattern of the segment polarity genes in {D}rosophila
  melanogaster,'' {\em Journal of Theoretical Biology}, vol.~223, no.~1,
  pp.~1--18, 2003.

\bibitem{davidich2008boolean}
M.~Davidich and S.~Bornholdt, ``{B}oolean network model predicts cell cycle
  sequence of fission yeast.,'' {\em Plos One}, vol.~3, no.~2, p.~e1672, 2008.

\bibitem{Cheng2010}
D.~Cheng, H.~Qi, and Z.~Li, {\em Analysis and Control of Boolean Networks: A
  Semi-tensor Product Approach}.
\newblock Springer Science \& Business Media, 2010.

\bibitem{Cheng2009controllability}
D.~Cheng and H.~Qi, ``Controllability and observability of {B}oolean control
  networks,'' {\em Automatica}, vol.~45, no.~7, pp.~1659--1667, 2009.

\bibitem{aa18}
Y.~Zhao, H.~Qi, and D.~Cheng, ``Input-state incidence matrix of {B}oolean
  control networks and its applications,'' {\em Systems \& Control Letters},
  vol.~59, no.~12, pp.~767--774, 2010.

\bibitem{Laschov2012}
D.~Laschov and M.~Margaliot, ``Controllability of {B}oolean control networks
  via the {P}erron-{F}robenius theory,'' {\em Automatica}, vol.~48, no.~6,
  pp.~1218--1223, 2012.

\bibitem{zhu2019further}
Q.~Zhu, Y.~Liu, J.~Lu, and J.~Cao, ``Further results on the controllability of
  {B}oolean control networks,'' {\em IEEE Transactions on Automatic Control},
  vol.~64, no.~1, pp.~440--442, 2019.

\bibitem{Fornasini2013}
E.~Fornasini and M.~Valcher, ``Observability, reconstructibility and state
  observers of {B}oolean control networks,'' {\em IEEE Transactions on
  Automatic Control}, vol.~58, no.~58, pp.~1390--1401, 2013.

\bibitem{cheng2016note}
D.~Cheng, H.~Qi, T.~Liu, and Y.~Wang, ``A note on observability of {B}oolean
  control networks,'' {\em Systems \& Control Letters}, vol.~87, pp.~76--82,
  2016.

\bibitem{Zhu2018OB}
Q.~Zhu, Y.~Liu, J.~Lu, and J.~Cao, ``Observability of {B}oolean control
  networks,'' {\em Science China Information Sciences}, vol.~61, no.~9,
  p.~092201, 2018.

\bibitem{guo2018redefined}
Y.~Guo, W.~Gui, and C.~Yang, ``Redefined observability matrix for {B}oolean
  networks and distinguishable partitions of state space,'' {\em Automatica},
  vol.~91, pp.~316--319, 2018.

\bibitem{chen2018monostability}
S.~Chen, Y.~Wu, M.~Macauley, and X.~Sun, ``Monostability and bistability of
  {B}oolean networks using semi-tensor products,'' {\em IEEE Transactions on
  Control of Network Systems, DOI: 10.1109/TCNS.2018.2889015}, 2018.

\bibitem{guo2018stability}
Y.~Guo, R.~Zhou, Y.~Wu, W.~Gui, and C.~Yang, ``Stability and set stability in
  distribution of probabilistic {B}oolean networks,'' {\em IEEE Transactions on
  Automatic Control, DOI: 10.1109/TAC.2018.2833170}, 2018.

\bibitem{meng2019stability}
M.~Meng, J.~Lam, J.~Feng, and K.~Cheung, ``Stability and stabilization of
  {B}oolean networks with stochastic delays,'' {\em IEEE Transactions on
  Automatic Control}, vol.~64, no.~2, pp.~790--796, 2019.

\bibitem{li2017lyapunov}
H.~Li and Y.~Wang, ``Lyapunov-based stability and construction of lyapunov
  functions for {B}oolean networks,'' {\em SIAM Journal on Control \&
  Optimization}, vol.~55, no.~6, pp.~3437--3457, 2017.

\bibitem{li2017stabilization}
F.~Li, H.~Li, L.~Xie, and Q.~Zhou, ``On stabilization and set stabilization of
  multivalued logical systems,'' {\em Automatica}, vol.~80, pp.~41--47, 2017.

\bibitem{aa9}
D.~Laschov and M.~Margaliot, ``A maximum principle for single-input {B}oolean
  control networks,'' {\em IEEE Transactions on Automatic Control}, vol.~56,
  no.~4, pp.~913--917, 2011.

\bibitem{aa10}
D.~Laschov and M.~Margaliot, ``Minimum-time control of {B}oolean networks,''
  {\em SIAM Journal on Control \& Optimization}, vol.~51, no.~4,
  pp.~2869--2892, 2012.

\bibitem{wu2019optimal}
Y.~Wu, X.~Sun, X.~Zhao, and T.~Shen, ``Optimal control of {B}oolean control
  networks with average cost: A policy iteration approach,'' {\em Automatica},
  vol.~100, pp.~378--387, 2019.

\bibitem{gao2018controllability}
Z.~Gao, X.~Chen, and T.~Ba{\c{s}}ar, ``Controllability of conjunctive {B}oolean
  networks with application to gene regulation,'' {\em IEEE Transactions on
  Control of Network Systems}, vol.~5, no.~2, pp.~770--781, 2018.

\bibitem{weiss2018minimal}
E.~Weiss, M.~Margaliot, and G.~Even, ``Minimal controllability of conjunctive
  {B}oolean networks is {NP}-complete,'' {\em Automatica}, vol.~92, pp.~56--62,
  2018.

\bibitem{weiss2018polynomial}
E.~Weiss and M.~Margaliot, ``A polynomial-time algorithm for solving the
  minimal observability problem in conjunctive {B}oolean networks,'' {\em IEEE
  Transactions on Automatic Control, DOI: 10.1109/TAC.2018.2882154}, 2018.

\bibitem{gao2018stability}
Z.~Gao, X.~Chen, and T.~Ba{\c{s}}ar, ``Stability structures of conjunctive
  {B}oolean networks,'' {\em Automatica}, vol.~89, pp.~8--20, 2018.

\bibitem{chen2017asymptotic}
X.~Chen, Z.~Gao, and T.~Ba{\c{s}}ar, ``Asymptotic behavior of conjunctive
  {{B}}oolean networks over weakly connected digraphs,'' {\em arXiv preprint
  arXiv:1708.01975}, 2017.

\bibitem{brualdi1991combinatorial}
R.~Brualdi and H.~Ryser, {\em Combinatorial Matrix Theory}.
\newblock Cambridge University Press, Cambridge, 1991.

\end{thebibliography}

\end{document}